\documentclass[epjST]{my_svjour}

\usepackage{amsfonts,amssymb,amsmath,cite,amsbsy,cases,empheq}

\usepackage{bm}

\newcommand{\PKi}{K_{P_{t_i}}}
\newcommand{\PBi}{B_{P_{t_i}}}

\newcommand{\PKid}{K_{P_{t_i}^*}}

\newcommand{\PAid}{A_{P_{t_i}^*}}
\newcommand{\N}{\mathbb{N}}
\newcommand{\R}{\mathbb{R}}
\newcommand{\oa}{\alpha_i}

\begin{document}

\title{Fractional Calculus of Variations\\ of Several Independent Variables}

\author{Tatiana Odzijewicz\inst{1}\fnmsep\thanks{\email{tatianao@ua.pt}}
\and Agnieszka B. Malinowska\inst{2}\fnmsep\thanks{\email{a.malinowska@pb.edu.pl}}
\and Delfim F. M. Torres\inst{1}\fnmsep\thanks{\email{delfim@ua.pt}}}

\institute{CIDMA --- Center for Research and Development in Mathematics and Applications,\\
Department of Mathematics, University of Aveiro, 3810-193 Aveiro, Portugal
\and Faculty of Computer Science, Bialystok University of Technology,
15-351 Bia\l ystok, Poland}


\abstract{We prove multidimensional integration by parts formulas
for generalized fractional derivatives and integrals. The new results
allow us to obtain optimality conditions for multidimensional
fractional variational problems with Lagrangians depending
on generalized partial integrals and derivatives. A generalized
fractional Noether's theorem, a formulation of Dirichlet's principle
and an uniqueness result are given.}

\maketitle


\markboth{T. Odzijewicz, A. B. Malinowska, D. F. M. Torres}{Fractional Calculus of Variations of Several Independent Variables}


\section{Introduction}

The research field concerned with extremal values
of functionals is called \emph{the calculus of variations}
\cite{book:Jurgen,book:vanBrunt}. Often, variational functionals
are given in the form of an integral that involves an unknown function
and its derivatives. In the simplest case, one thinks of single variable integration.
However, results can be further extended to the multi-time calculus.
Variational problems are particularly attractive because of their many-fold applications.
For example, in physics, engineering, and economics, the variational integral
may represent some action, energy, or cost functional \cite{book:Ewing,book:Weinstock}.
The calculus of variations possesses also important connections with other fields of mathematics.
Here we are interested in connections with fractional calculus,
which is a generalization of the standard calculus
that considers integrals and derivatives of noninteger (real or complex) order
\cite{book:Kilbas,book:Podlubny,book:Samko}. The first question linking the two areas
was brought up in the XIXth century by Niels Heinrik Abel (1802--1829).
Abel's mechanical problem asks about a curve, lying in a vertical plane, for which the time taken
by a material point sliding without friction from the highest point to the lowest one,
is destined function of time \cite{Abel}. The problem
is a generalization of the tautochrone problem, which is part of the calculus of variations.
Nevertheless, only in 1996--1997, with the works of Riewe \cite{CD:Riewe:1996,CD:Riewe:1997},
the fractional calculus of variations became an important research field \emph{per se}
\cite{B:O:T,MyID:207,book:Klimek,MyID:249,Cresson,DerInt,gastao1,gastao4,OMT:2013}.
It is nowadays of strong interest, with many authors contributing to its theory and applications.
For the state of the art we refer the reader to the recent book \cite{book:fcv}.

Our goal is to develop a theory of the fractional calculus of variations
by considering multidimensional fractional variational problems
with Lagrangians depending on generalized partial fractional operators. Moreover,
applications to physics are discussed (see Example~\ref{ex:1},
Sections~\ref{sec:DP} and \ref{sec:NT}).
Our results generalize the fractional calculus
of variations for functionals involving multiple integrals
studied in \cite{MyID:182,tatiana,tarasov}, as well as previous works about
extremizers of single variable integral functionals
with generalized fractional operators \cite{OmPrakashAgrawal,FVC_Gen_Int,FVC_Gen}.

The text is organized as follows. In Section~\ref{sec:prelim}
we give definitions and basic properties for the generalized ordinary
and partial fractional operators. Main results are then proved and discussed
in Sections~\ref{sec:MR} and \ref{sec:MFCV}:
multidimensional fractional integration by parts formulas
are given in Section~\ref{sec:MR} (Theorems~\ref{theorem:IPRI} and \ref{theorem:IBPD});
while in Section~\ref{sec:MFCV} we obtain fractional partial differential equations
of the Euler--Lagrange type for multi-time variational problems with a Lagrangian depending
on generalized partial fractional operators (Theorems~\ref{thm:ELCaputo}
and \ref{thm:EL:classical:Caputo}), we prove
a generalized fractional Dirichlet's principle
(Theorems~\ref{thm:GFDP} and \ref{thm:unc:sol}),
and a fractional Noether's symmetry theorem (Theorem~\ref{noether}).
We end with Section~\ref{sec:conc} of conclusion.


\section{Generalized Fractional Operators}
\label{sec:prelim}

In this section we give definitions of generalized ordinary and partial fractional operators.
By the choice of an appropriate kernel, these operators can be reduced to the standard
fractional integrals and derivatives. For more on the subject we refer the reader to
\cite{OmPrakashAgrawal,book:Kiryakova,FVC_Gen_Int,FVC_Gen}.

\begin{definition}[Generalized fractional integral]
\label{def:GI}
Let $f:[a,b]\rightarrow \mathbb{R}$. The operator $K_P^\alpha$ is defined by
\begin{equation*}
K_P^{\alpha}f(t)
:=p\int\limits_{a}^{t}k_{\alpha}(t,\tau)f(\tau)d\tau
+q\int\limits_{t}^{b}k_{\alpha}(\tau,t)f(\tau)d\tau,
\end{equation*}
where $P=\langle a,t,b,p,q\rangle$ is the \emph{parameter set} ($p$-set for brevity),
$t\in(a,b)$, $p,q$ are real numbers, and $k_{\alpha}(t,\tau)$
is a kernel which may depend on $\alpha$.
The operator $K_P^\alpha$ is referred as the \emph{operator $K$}
($K$-op for simplicity) of order $\alpha$ and $p$-set $P$.
\end{definition}

\begin{theorem}[cf. Theorem~2.3 of \cite{FVC_Gen}]
\label{theorem:L1}
Let $k_\alpha$ be a difference kernel, i.e., $k_\alpha(t,\tau)=k_\alpha(t-\tau)$
and $k_\alpha\in L_1\left((0,b-a);\R\right)$. Then
$K_P^{\alpha}:L_1\left((a,b);\R\right)\rightarrow L_1\left((a,b);\R\right)$
is well defined, bounded and linear operator.
\end{theorem}

\begin{definition}[Generalized Riemann--Liouville fractional derivative]
\label{def:GRL}
Let $P$ be a given parameter set. The operator $A_P^\alpha$,
$0 < \alpha < 1$, is defined by
\begin{equation*}
A_P^\alpha := D \circ K_P^{1-\alpha},
\end{equation*}
where $D$ denotes the standard derivative.
We refer to $A_P^\alpha$ as \emph{operator $A$} ($A$-op)
of order $\alpha$ and $p$-set $P$.
\end{definition}

\begin{definition}[Generalized Caputo fractional derivative]
\label{def:GC}
Let $P$ be a given parameter set. The operator $B_P^\alpha$,
$\alpha \in (0,1)$, is defined by
\begin{equation*}
B_P^\alpha :=K_P^{1-\alpha} \circ D
\end{equation*}
and is referred as the \emph{operator $B$} ($B$-op)
of order $\alpha$ and $p$-set $P$.
\end{definition}

Operators $A$, $B$ and $K$ reduce to the classical
fractional integrals and derivatives for suitably
chosen kernels and $p$-sets (see \cite{FVC_Gen_Int,FVC_Gen}).
The notation was introduced in \cite{OmPrakashAgrawal}
and is now standard \cite{Lupa,OMT:2013}.

Now, we shall define generalized partial fractional operators.
For $n\in\mathbb{N}$, let $\bm{\alpha}=(\alpha_1,\dots,\alpha_n)$,
$p=(p_1,\dots,p_n)$, $q=(q_1,\dots,q_n)\in\mathbb{R}^n$
with $0<\alpha_i<1$, $i=1,\dots,n$, and
$\Delta_n=(a_1,b_1)\times\dots\times (a_n,b_n)\subset \mathbb{R}^n $,
$t=(t_1,\dots, t_n)\in \Delta_n$. Generalized partial fractional integrals
and derivatives are natural generalizations of the corresponding one-dimensional
generalized fractional integrals and derivatives,
being taken with respect to one or several variables.

\begin{definition}[Generalized partial fractional integral]
\label{def:GPI}
Let  $f=f(t_1,\dots,t_n):\bar{\Delta}_n\rightarrow \mathbb{R}$.
The generalized partial Riemann--Liouville fractional integral
of order $\alpha_i$ with respect to the $i$th variable $t_i$ is given by
\begin{multline*}
K_{P_{t_i}}^{\alpha_i}f(t)
:=p_i\int\limits_{a_i}^{t_i}k_{\alpha_i}(t_i,\tau)
f(t_1,\dots,t_{i-1},\tau,t_{i+1},\dots,t_n)d\tau \\
+q_i\int\limits_{t_i}^{b_i}k_{\alpha_i}(\tau,t_i)
f(t_1,\dots,t_{i-1},\tau,t_{i+1},\dots,t_n)d\tau,
\end{multline*}
where $P_{t_i}=\langle a_i,t_i,b_i,p_i,q_i \rangle $. We refer to
$K_{P_{t_i}}^{\alpha}$ as the \emph{partial operator $K$} (partial $K$-op)
of order $\alpha_i$ and $p$-set $P_{t_i}$.
\end{definition}

\begin{definition}[Generalized partial Riemann--Liouville fractional derivative]
\label{def:GPRL}
Let $P_{t_i}=\langle a_i,t_i,b_i,p_i,q_i \rangle $. The generalized partial
Riemann--Liouville fractional derivative of order $\alpha$
with respect to the $i$th variable $t_i$ is given by
\begin{multline*}
A_{P_{t_i}}^{\alpha_i}f(t):=\frac{\partial}{\partial t_i}\left(
p_i\int\limits_{a_i}^{t_i}k_{1-\alpha_i}(t_i,\tau)
f(t_1,\dots,t_{i-1},\tau,t_{i+1},\dots,t_n)d\tau \right. \\
\left.+q_i\int\limits_{t_i}^{b_i}k_{1-\alpha_i}(\tau,t_i)
f(t_1,\dots,t_{i-1},\tau,t_{i+1},\dots,t_n)d\tau\right)
=\left(\frac{\partial}{\partial t_i}K_{P_{t_i}}^{1-\alpha_i}f\right)(t) .
\end{multline*}
The operator $A_{P_{t_i}}^{\alpha_i}$ is referred as the \emph{partial operator $A$}
(partial $A$-op) of order $\alpha_i$ and $p$-set $P_{t_i}$.
\end{definition}

\begin{definition}[Generalized partial Caputo fractional derivative]
\label{def:GPC}
Let  $P_{t_i}=\langle a_i,t_i,b_i,p_i,q_i \rangle $. The generalized partial Caputo
fractional derivative of order $\alpha_i$ with respect to the $i$th variable $t_i$ is given by
\begin{multline*}
B_{P_{t_i}}^{\alpha_i}f(t)
:=p_i\int\limits_{a_i}^{t_i}k_{1-\alpha_i}(t_i,\tau)\frac{\partial}{\partial \tau}
f(t_1,\dots,t_{i-1},\tau,t_{i+1},\dots,t_n)d\tau \\
+q_i\int\limits_{t_i}^{b_i}k_{1-\alpha_i}(\tau,t_i)\frac{\partial}{\partial \tau}
f(t_1,\dots,t_{i-1},\tau,t_{i+1},\dots,t_n)d\tau
=\left(K_{P_{t_i}}^{1-\alpha_i} \frac{\partial}{\partial t_i} f\right)(t)
\end{multline*}
and is referred as the \emph{partial operator $B$} (partial $B$-op)
of order $\alpha_i$ and $p$-set $P_{t_i}$.
\end{definition}

Similarly to the one-dimension case, partial operators $K$, $A$ and $B$ reduce
to the standard partial fractional integrals and derivatives. The left- and right-sided
Riemann--Liouville partial fractional integrals with respect to the $i$th variable
$t_i$ are obtained by choosing the kernel $k_{\alpha_i}(t_i,\tau)
=\frac{1}{\Gamma(\alpha_i)}(t_i-\tau)^{\alpha_i-1}$ and $p$-sets
$P_{t_i}=\langle a_i,t_i,b_i,1,0\rangle$ and
$P_{t_i}=\langle a_i,t_i,b_i,0,1\rangle$, respectively:
\begin{equation*}
K_{P_{t_i}}^{\alpha_i}f(t)=\frac{1}{\Gamma(\alpha_i)}
\int\limits_{a_i}^{t_i}(t_i-\tau)^{\alpha_i-1}f(t_1,\dots,t_{i-1},\tau,t_{i+1},
\dots,t_n)d\tau=: {_{a_i}}\textsl{I}^{\alpha_i}_{t_i} f(t),
\end{equation*}
\begin{equation*}
K_{P_{t_i}}^{\alpha}f(t)=\frac{1}{\Gamma(\alpha_i)}\int\limits_{t_i}^{b_i}(\tau-t_i)^{\alpha_i-1}
f(t_1,\dots,t_{i-1},\tau,t_{i+1},\dots,t_n)d\tau=: {_{t_i}}\textsl{I}^{\alpha_i}_{b_i} f(t).
\end{equation*}
The standard left- and right-sided Riemann--Liouville and Caputo partial fractional derivatives
with respect to the $i$th variable $t_i$ are received by choosing the kernel
$k_{1-\alpha_i}(t_i,\tau)=\frac{1}{\Gamma(1-\alpha_i)}(t_i-\tau)^{-\alpha_i}$:
if $P_{t_i}=\langle a_i,t_i,b_i,1,0\rangle$, then
\begin{equation*}
A_{P_{t_i}}^{\alpha_i}f(t)=\frac{1}{\Gamma(1-\alpha_i)}\frac{\partial}{\partial t_i}
\int\limits_{a_i}^{t_i}(t_i-\tau)^{-\alpha_i}f(t_1,\dots,t_{i-1},\tau,t_{i+1},\dots,t_n)d\tau
=:{_{a_i}}\textsl{D}^{\alpha_i}_{t_i} f(t),
\end{equation*}
\begin{equation*}
B_{P_{t_i}}^{\alpha_i}f(t)=\frac{1}{\Gamma(1-\alpha_i)}\int\limits_{a_i}^{t_i}
(t_i-\tau)^{-\alpha_i}\frac{\partial}{\partial \tau}
f(t_1,\dots,t_{i-1},\tau,t_{i+1},\dots,t_n)d\tau\\
=:{^{C}_{a_i}}\textsl{D}^{\alpha_i}_{t_i} f(t);
\end{equation*}
if $P_{t_i}=\langle a_i,t_i,b_i,0,1\rangle$, then
\begin{equation*}
-A_{P_{t_i}}^{\alpha_i}f(t)=\frac{1}{\Gamma(1-\alpha_i)}\frac{\partial}{\partial t_i}
\int\limits_{t_i}^{b_i}(\tau-t_i)^{-\alpha_i}f(t_1,\dots,t_{i-1},\tau,t_{i+1},\dots,t_n)d\tau\\
=:{_{t_i}}\textsl{D}^{\alpha_i}_{b_i} f(t),
\end{equation*}
\begin{equation*}
-B_{P_{t_i}}^{\alpha}f(t)=\frac{1}{\Gamma(1-\alpha_i)}
\int\limits_{t_i}^{b_i}(\tau-t_i)^{-\alpha_i}\frac{\partial}{\partial \tau}
f(t_1,\dots,t_{i-1},\tau,t_{i+1},\dots,t_n)d\tau \\
=:{^{C}_{t_i}}\textsl{D}^{\alpha_i}_{b_i} f(t).
\end{equation*}

\begin{remark}
In Definitions~\ref{def:GPI}, \ref{def:GPRL} and \ref{def:GPC}
all the variables except $t_i$ are kept fixed. That choice
of fixed values determines a function
$f_{t_1,\dots,t_{i-1},t_{i+1},\dots,t_n}:(a_i,b_i)\rightarrow \mathbb{R}$
of one variable $t_i$:
$f_{t_1,\dots,t_{i-1},\dots,t_{i+1},\dots,t_n}(t_i)
=f(t_1,\dots,t_{i-1},t_i,t_{i+1},\dots,t_n)$.
By Definitions~\ref{def:GI}, \ref{def:GRL}, \ref{def:GC}
and \ref{def:GPI}, \ref{def:GPRL}, \ref{def:GPC} we have
$$
K_{P_{t_i}}^{\alpha_i}f_{t_1,\dots,t_{i-1},t_{i+1},\dots,t_n}(t_i)
=K_{P_{t_i}}^{\alpha_i}f(t_1,\dots,t_{i-1},t_i,t_{i+1},\dots,t_n),
$$
$$
A_{P_{t_i}}^{\alpha_i}f_{t_1,\dots,t_{i-1},t_{i+1},\dots,t_n}(t_i)
=A_{P_{t_i}}^{\alpha_i}f(t_1,\dots,t_{i-1},t_i,t_{i+1},\dots,t_n),
$$
$$
B_{P_{t_i}}^{\alpha_i}f_{t_1,\dots,t_{i-1},t_{i+1},\dots,t_n}(t_i)
=B_{P_{t_i}}^{\alpha_i}f(t_1,\dots,t_{i-1},t_i,t_{i+1},\dots,t_n).
$$
Therefore, as in the integer-order case, computation of partial
generalized fractional operators reduces to the computation
of one-variable generalized fractional operators.
\end{remark}


\section{Generalized Fractional Integration by Parts
for Functions of Several Variables}
\label{sec:MR}

The integration by parts formula plays a crucial role
in the principle of virtual works. In this section,
it is of our interest to obtain such formula for generalized
fractional operators. Throughout the section,
$i\in\{1,\ldots,n\}$ is arbitrary but fixed.

\begin{definition}[Dual $p$-set]
For a given $p$-set $P_{t_i}=\langle a_i,t_i,b_i,p_i,q_i\rangle$
we denote by $P_{t_i}^{*}$ the $p$-set
$P_{t_i}^{*} = \langle  a_i,t_i,b_i,q_i,p_i\rangle$.
We say that $P_{t_i}^{*}$ is the dual of $P_{t_i}$.
\end{definition}

\begin{theorem}
\label{theorem:IPRI}
Let $\alpha_i\in (0,1)$ and $P_{t_i}=\langle a_i,t_i,b_i,p_i,q_i\rangle$
be a parameter set. Moreover, let $k_{\alpha_i}$ be a difference kernel,
i.e., $k_{\alpha_i}(t_i,\tau)=k_{\alpha_i}(t_i-\tau)$ such that
$k_{\alpha_i}\in L_1\left((0,b_i-a_i);\R\right)$. If $f:\R^n\rightarrow\R$
and $\eta:\R^n\rightarrow\R$, $f,\eta\in C\left(\bar{\Delta}_n;\R\right)$,
then the generalized partial fractional integrals satisfy the identity
\begin{equation*}
\int\limits_{\Delta_n}f\cdot \PKi^{\alpha_i}\eta~dt_n\dots dt_1
=\int\limits_{\Delta_n} \eta\cdot \PKid^{\alpha_i} f~dt_n\dots dt_1,
\end{equation*}
where $P_{t_i}^*$ is the dual of $P_{t_i}$.
\end{theorem}

\begin{proof}
Let $\alpha_i\in (0,1)$, $P_{t_i}=\langle a_i,t_i,b_i,p_i,q_i \rangle$
and $f,\eta,\in C\left(\bar{\Delta}_n;\R\right)$. Define
\[
F(t,\tau):=
\left\{
\begin{array}{ll}
\left|p_i k_{\alpha_i}(t_i-\tau)\right|
\cdot\left|f(t)\right|\cdot \left|\eta(t_1,\dots,t_{i-1},\tau,t_{i+1},\dots,t_n)\right|
& \mbox{if $\tau < t_i$}\\
\left|q_i k_{\alpha_i}(\tau-t_i)\right|
\cdot \left|f(t)\right|\cdot\left|\eta(t_1,\dots,t_{i-1},\tau,t_{i+1},\dots,t_n)\right|
& \mbox{if $\tau > t_i$}
\end{array}\right.
\]
for all $(t,\tau)\in \Delta_n\times (a_i,b_i)$.
Since $f$ and $\eta$ are continuous functions on $\bar{\Delta}_n$,
they are bounded on $\bar{\Delta}_n$, \textrm{i.e.}, there exist real numbers
$C, D>0$ such that $\left|f(t)\right|\leq C$ and $\left|\eta(t)\right| \leq D$
for all $t\in \Delta_n$. Therefore,
\begin{equation*}
\begin{split}
\int\limits_{\Delta_n} &\left(\int_{a_i}^{b_i} F(t,\tau) d\tau \right)dt_n\dots dt_1\\
&=\int\limits_{\Delta_n}\left(\int_{a_i}^{t_i} \left|p_i k_{\alpha_i}(t_i-\tau)\right|
\cdot\left|f(t)\right|\cdot \left|\eta(t_1,
\dots,t_{i-1},\tau,t_{i+1},\dots,t_n)\right| d\tau\right. \\
& \quad \left.+\int_{t_i}^{b_i}\left|q_ik_{\alpha_i}(\tau-t_i)\right|
\cdot \left|f(t)\right|\cdot\left|\eta(t_1,\dots,t_{i-1},\tau,t_{i+1},
\dots,t_n)\right| d\tau \right)dt_n\dots dt_1 \\
&\leq C\cdot D\int\limits_{\Delta_n}\left(\int_{a_i}^{t_i}
\left| p_ik_{\alpha_i}(t_i-\tau)\right| d\tau
+ \int_{t_i}^{b_i}\left|q_ik_{\alpha_i}(\tau-t_i)\right|d\tau \right)dt_n\dots dt_1 \\
&\leq C\cdot D \left(\left|q_i\right|+\left|p_i\right|\right)\left\|k_{\alpha_i}\right\|
\cdot\prod_{i=1}^n(b_i-a_i)\\
&<\infty.
\end{split}
\end{equation*}
The result follows by using Fubini's theorem to change
the order of integration in the iterated integrals:
\begin{equation*}
\begin{split}
\int_{\Delta_n} & f\cdot \PKi^{\oa}\eta dt_n\dots dt_1\\
&=\int\limits_{\Delta_n}\left( p_i\int_{a_i}^{t_i} f(t)k_{\alpha_i}(t_i-\tau)
\eta(t_1,\dots,t_{i-1},\tau,t_{i+1},\dots,t_n)d\tau \right.\\
&\  \left.+q_i \int_{t_i}^{b_i}f(t)k_{\alpha_i}(\tau-t_i)\eta(t_1,
\dots,t_{i-1},\tau,t_{i+1},\dots,t_n)d\tau \right)dt_n\dots dt_1\\
&=\int\limits_{\Delta_n}\left(p_i\int_{\tau}^{b_i} f(t)k_{\alpha_i}(t_i-\tau)
\eta(t_1,\dots,t_{i-1},\tau,t_{i+1},\dots,t_n) dt_i \right. + q_i \\
&\left. \ \times \int_{a_i}^{\tau}f(t)k_{\alpha_i}(\tau-t_i)\eta(t_1,
\dots,t_{i-1},\tau,t_{i+1},\dots,t_n)dt_i \right)dt_n\dots dt_{i-1}d\tau dt_{i+1}\dots dt_1\\
&=\int\limits_{\Delta_n}\eta(t_1,\dots,t_{i-1},\tau,t_{i+1},\dots,t_n)\left(
p_i\int_{\tau}^{b_i} f(t)k_{\alpha_i}(t_i-\tau) dt_i \right. \\
& \ \left.+q_i\int_{a_i}^{\tau}f(t)k_{\alpha_i}(\tau-t_i)dt_i \right)dt_n
\dots dt_{i-1}d\tau dt_{i+1}\dots dt_1\\
&=\int\limits_{\Delta_n}\eta\cdot K_{P_{t_i}^*}^{\alpha_i} f dt_n\dots dt_1.
\end{split}
\end{equation*}
\end{proof}

\begin{theorem}[Generalized fractional integration by parts]
\label{theorem:IBPD}
Let $\oa\in (0,1)$ and $P_{t_i}=\langle a_i,t_i,b_i,p_i,q_i\rangle$ be a parameter
set and $f,\eta\in C^1\left(\bar{\Delta}_n;\R\right)$. Moreover, let $k_{\alpha_i}$
be a difference kernel such that $k_{1-\alpha_i}\in L_1\left((0,b_i-a_i);\R\right)$
and $K_{P_{t_i}^*}^{1-\alpha_i} f\in C^1\left(\bar{\Delta}_n;\R\right)$.
Then
\begin{equation*}
\int\limits_{\Delta_n} f\cdot \PBi^{\oa}\eta~dt_n\dots dt_1
=\int\limits_{\partial\Delta_n} \eta\cdot\PKid^{1-\oa} f\cdot\nu^i~d(\partial\Delta_n)
-\int\limits_{\Delta_n}\eta\cdot \PAid^{\oa} f~dt_n\dots dt_1,
\end{equation*}
where $\nu^i$ is the outward pointing unit normal to $\partial\Delta_n$.
\end{theorem}

\begin{proof}
By definition of the generalized partial Caputo fractional derivative,
Theorem~\ref{theorem:IPRI}, and the standard integration by parts formula
(see, e.g.,\cite{book:Evans}), one has
\begin{equation*}
\begin{split}
\int\limits_{\Delta_n} f &\cdot \PBi^{\oa}\eta~dt_n\dots dt_1\\
&=\int\limits_{\Delta_n} f\cdot K_{P_{t_i}}^{1-\alpha_i}
\frac{\partial\eta}{\partial t_i}~dt_n\dots dt_1
=\int\limits_{\Delta_n}\frac{\partial\eta}{\partial t_i}
K_{P_{t_i}^*}^{1-\alpha_i}f~dt_n\dots dt_1\\
&=\int\limits_{\partial\Delta_n}\eta\cdot\PKid^{1-\oa}
f\cdot\nu^i~d(\partial\Delta_n)-\int\limits_{\Delta_n} \eta
\cdot\frac{\partial}{\partial t_i} K_{P_{t_i}^*}^{1-\alpha_i}f~dt_n\dots dt_1\\
&=\int\limits_{\partial\Delta_n}\eta\cdot\PKid^{1-\oa}f\cdot\nu^i~d(\partial\Delta_n)
-\int\limits_{\Delta_n} \eta\cdot\PAid^{\alpha_i}f~dt_n\dots dt_1.
\end{split}
\end{equation*}
\end{proof}


\section{The Generalized Fractional Calculus of Variations of Several Independent Variables}
\label{sec:MFCV}

Variational problems with functionals depending on several independent variables arise,
for example, in mechanics, for systems with infinite number of degrees of freedom,
like a vibrating elastic solid. Fractional variational problems involving multiple
integrals have been already studied in different contexts. We can mention here
\cite{MyID:182,Baleanu,Cresson,tatiana}, where the multidimensional
fractional Euler--Lagrange equations for the field are obtained,
or \cite{malinowska:1,malinowska:2}, where a first and a second fractional
Noether-type theorem are proved. In this section we present a more general approach
to the subject by considering functionals depending on generalized fractional operators.
In the sequel we use the notion of generalized fractional gradient.

\begin{definition}[The generalized fractional gradient operator]
Let $n\in\N$, $P=\left(P_{t_1},\dots, P_{t_n}\right)$, and $\bm{\alpha}\in (0,1)^n$.
We define the generalized fractional gradient of a function
$f:\mathbb{R}^n\rightarrow\mathbb{R}$ with respect
to a generalized fractional operator $T$ by
\begin{equation*}
\nabla_{T_{P}}^{\bm{\alpha}}f:=\sum_{i=1}^n e_i T_{P_{t_i}}^{\alpha_i}f,
\end{equation*}
where $\left\{e_i:i=1,\dots n\right\}$ denotes the standard basis in $\mathbb{R}^n$.
Additionally, we define $\nabla_{T_{P}}^{\bm{\alpha}}f$ for a vector function
$f:\mathbb{R}^n\rightarrow\mathbb{R}^N$ by
$$
\nabla_{T_{P}}^{\bm{\alpha}} f
:= \left[\nabla_{T_{P}}^{\bm{\alpha}}f_1,
\ldots,\nabla_{T_{P}}^{\bm{\alpha}}f_N\right].
$$
\end{definition}


\subsection{The Fundamental Problem}

Let $\bm{\alpha}=(\alpha_1,\dots,\alpha_n)$ and $\bm{\beta}=(\beta_1,\dots,\beta_n)$
be such that $\alpha_i,\beta_i \in(0,1)$, and $P^{j}=\left(P^{j}_{t_1},\dots,P^{j}_{t_n}\right)$,
where $P_{t_i}^j=\langle a_i,t_i,b_i,p_i^j,q_i^j\rangle$, $i=1,\dots,n$, $j=1,2$.
Consider the problem of finding an extremizer
$u:\Delta_n\rightarrow\R^N$ of the functional
\begin{equation}
\label{eq:FundFunct}
\mathcal{J}[u]=\int\limits_{\Delta_n}
F\left(t,u(t),\nabla_{B_{P^1}}^{\bm{\alpha}}u(t),
\nabla_{K_{P^2}}^{\bm{\beta}}u(t)\right)dt_n\dots dt_1
\end{equation}
subject to the boundary condition
\begin{equation}
\label{eq:FundBound}
\left.u(t)\right|_{\partial \Delta_n}\equiv\psi(t),
\end{equation}
where $\psi:\partial\Delta_n\rightarrow \mathbb{R}^N$ is a given function.
For simplicity of notation we write
\begin{equation*}
\left\{u\right\}_{P^1, P^2}^{\bm{\alpha},\bm{\beta}}(t)
=\left(t,u(t),\nabla_{B_{P^1}}^{\alpha}u(t),\nabla_{K_{P^2}}^{\beta}u(t)\right)
\end{equation*}
and
$dt=dt_n\dots dt_1$. As usually, we denote by $\partial_i F$, $i=1,\ldots,M$ ($M\in\mathbb{N}$),
the partial derivative of function $F:\R^M\rightarrow \R$ with
respect to its $i$th argument. We assume that
$F\in C^2\left(\Delta_n\times\mathbb{R}^N\times\mathbb{R}^{2nN};\mathbb{R}\right)$,
$t \mapsto \partial_{N+kn+i} F\left\{u\right\}_{P^1, P^2}^{\bm{\alpha},\bm{\beta}}(t)$
has continuously differentiable partial integral $K_{P_{t_i}^{1*}}^{1-\alpha_i}$
and continuous partial derivative $A_{P_{t_i}^{1*}}^{\alpha_i}$; and
$t \mapsto \partial_{n+N(k+n)+i} F\left\{u\right\}_{P^1, P^2}^{\bm{\alpha},\bm{\beta}}(t)$
has continuous partial integral $K_{P_{t_i}^{2*}}^{\beta_i}$, where $i=1,\dots,n$ and $k=1\dots,N$.
Moreover, we suppose that $k_{\alpha_i}$ and $k_{\beta_i}$ are difference kernels such that
$k_{1-\alpha_i}$, $k_{\beta_i}\in L_1\left((0,b_i-a_i);\R\right)$, $i=1,\dots,n$.

\begin{definition}
A continuously differentiable function
$u\in C^1\left(\bar{\Delta}_n;\mathbb{R}^N\right)$ is said to be
admissible for the variational problem \eqref{eq:FundFunct}--\eqref{eq:FundBound}
if, for all $i\in\{1,\dots,n\}$, $B_{P_{t_i}^{1}}^{\alpha_i}u$
and $K_{P_{t_i}^{2}}^{\beta_i}u$ exist and are continuous on $\bar{\Delta}_n$
and $u$ satisfies the boundary condition \eqref{eq:FundBound}.
\end{definition}

\begin{theorem}
\label{thm:ELCaputo}
Let $u$ be a solution to problem \eqref{eq:FundFunct}--\eqref{eq:FundBound}.
Then, $u$ satisfies the following system of fractional partial differential equations
of Euler--Lagrange type:
\begin{multline}
\label{eq:eqELCaputo}
\sum\limits_{i=1}^{n} \left[-A_{P_{t_i}^{1*}}^{\alpha_i}\partial_{N+kn+i}
F\left\{u\right\}_{P^1, P^2}^{\bm{\alpha},\bm{\beta}}(t)
+K_{P_{t_i}^{2*}}^{\beta_i}\partial_{n+N(k+n)+i}
F\left\{u\right\}_{P^1, P^2}^{\bm{\alpha},\bm{\beta}}(t)\right]\\
+\partial_{n+k} F\left\{u\right\}_{P^1, P^2}^{\bm{\alpha},\bm{\beta}}(t)=0,
\end{multline}
$k=1,\dots,N$, for all $t\in\Delta_n$.
\end{theorem}

\begin{proof}
Suppose that $u$ is an extremizer of $\mathcal{J}$.
For $\eta\in C^1\left(\bar{\Delta}_n;\mathbb{R}^N\right)$
such that $B_{P_{t_i}^{1}}^{\alpha_i}\eta$ and $K_{P_{t_i}^{2}}^{\beta_i}\eta$
are continuous  for all $i\in\{1,\dots,n\}$, and
$\left.\eta(t)\right|_{\partial \Delta_n}\equiv0$, $\varepsilon\in\mathbb{R}$,
the function $\hat{u}(t)=u(t)+\varepsilon\eta(t)$ is still admissible. Define
$$
J(\varepsilon):=\mathcal{J}[\hat{u}]
=\int\limits_{\Delta_n} F\left(t,\hat{u}(t),\nabla_{B_{P^1}}^{\bm{\alpha}}
\hat{u}(t),\nabla_{K_{P^2}}^{\bm{\beta}}\hat{u}(t)\right)dt.
$$
Then, a necessary condition for $u$ to be an extremizer of $J$ is given by
$\left.\frac{d J}{d\varepsilon}\right|_{\varepsilon=0}=0$, that is,
\begin{multline*}
\int\limits_{\Delta_n} \sum\limits_{k=1}^{N}\Biggl(\partial_{n+k}
F\left\{u\right\}_{P^1, P^2}^{\bm{\alpha},\bm{\beta}}(t)\cdot \eta_k(t)
+\sum\limits_{i=1}^{n} \left[\partial_{N+kn+i} F\left\{u\right\}_{P^1, P^2}^{\bm{\alpha},
\bm{\beta}}(t)B_{P_{t_i}^{1}}^{\alpha_i}\eta_k(t)\right.\\
\left.+\partial_{n+N(k+n)+i} F\left\{u\right\}_{P^1, P^2}^{\bm{\alpha},
\bm{\beta}}(t)K_{P_{t_i}^{2}}^{\beta_i}\eta_k(t)\right]\Biggr)dt = 0.
\end{multline*}
By integration by parts formulas (Theorems~\ref{theorem:IPRI} and \ref{theorem:IBPD})
and since $\left.\eta(t)\right|_{\partial \Delta_n}\equiv0$, one has
\begin{equation*}
\int\limits_{\Delta_n}\left(\partial_{N+kn+i} F\cdot B_{P_{t_i}^{1}}^{\alpha_i}
\eta_k \right)dt=-\int\limits_{\Delta_n}\eta_k
\cdot\left(A_{P_{t_i}^{1*}}^{\alpha_i}\partial_{N+kn+i} F\right) dt,
\end{equation*}
\begin{equation*}
\int\limits_{\Delta_n}\left(\partial_{n+N(k+n)+i} F\cdot K_{P_{t_i}^{2}}^{\beta_i}\eta_k\right) dt
=\int\limits_{\Delta_n}\eta_k\cdot\left(K_{P_{t_i}^{2*}}^{\beta_i}\partial_{n+N(k+n)+i} F \right)dt,
\end{equation*}
where $i=1,\dots,n$ and $k=1,\dots,N$. Therefore,
\begin{multline*}
\int\limits_{\Delta_n}\sum\limits_{k=1}^{N}\eta_k(t)\cdot\Biggl(\partial_{n+k}
F\left\{u\right\}_{P^1, P^2}^{\bm{\alpha},\bm{\beta}}(t)
+\sum\limits_{i=1}^{n} \left[-A_{P_{t_i}^{1*}}^{\alpha_i}\partial_{N+kn+i}
F\left\{u\right\}_{P^1, P^2}^{\bm{\alpha},\bm{\beta}}(t)\right.\\
\left.+K_{P_{t_i}^{2*}}^{\beta_i}\partial_{n+N(k+n)+i}
F\left\{u\right\}_{P^1, P^2}^{\bm{\alpha},\bm{\beta}}(t)\right]\Biggr)dt = 0.
\end{multline*}
Finally, by the fundamental lemma of the calculus of variations,
we arrive to \eqref{eq:eqELCaputo}.
\end{proof}

\begin{definition}
We say that an admissible function $u$ is an extremal for problem
\eqref{eq:FundFunct}--\eqref{eq:FundBound} if it satisfies the
system of fractional partial differential equations \eqref{eq:eqELCaputo}.
\end{definition}

Using similar techniques as in the proof of Theorem~\ref{thm:ELCaputo},
one can prove the following theorem.

\begin{theorem}
\label{thm:EL:classical:Caputo}
Let $u:\Delta_n\rightarrow\R^N$ be an extremizer of
\begin{equation*}
\mathcal{J}[u]=\int\limits_{\Delta_n} F\left(t,u(t),
\nabla_{B_{P^1}}^{\bm{\alpha}}u(t),\nabla u(t)\right)dt_n\dots dt_1
\end{equation*}
subject to the boundary condition $\left.u(t)\right|_{\partial \Delta_n}\equiv\psi(t)$,
where $\psi:\partial\Delta_n\rightarrow \mathbb{R}^N$ is a given function.
Then, $u$ satisfies the system of multidimensional generalized Euler--Lagrange equations
\begin{multline*}
\sum\limits_{i=1}^{n} \left[A_{P_{t_i}^{1*}}^{\alpha_i}\partial_{N+kn+i}
F\left\{u\right\}_{P^1, P^2}^{\bm{\alpha},\bm{\beta}}(t)
+\frac{\partial}{\partial_{t_i}}\partial_{n+N(k+n)+i}
F\left\{u\right\}_{P^1, P^2}^{\bm{\alpha},\bm{\beta}}(t)\right]\\
=\partial_{n+k} F\left\{u\right\}_{P^1, P^2}^{\bm{\alpha},\bm{\beta}}(t),
\end{multline*}
$k=1,\dots,N$, for all $t\in\Delta_n$.
\end{theorem}

\begin{example}
\label{ex:1}
Consider a medium motion whose displacement may be described by a scalar function
$u(t,x)$, where $x=(x_1,x_2)$. For example, this function may represent
the transverse displacement of a membrane. Suppose that the kinetic energy
$T$ and the potential energy $V$ of the medium are given by
$T\left(\frac{\partial u}{\partial t}\right)=\frac{1}{2}\int \rho
\left(\frac{\partial u}{\partial t}\right)^2 dx$ and $V(u)=\frac{1}{2}
\int k |\nabla u |^2dx$, respectively, where $\rho (x)$ is the mass density
and $k(x)$ is the stiffness, both assumed positive. Then, the classical action functional is
$\mathcal{J}(u)=\frac{1}{2}\int \int\left(\rho
\left(\frac{\partial u}{\partial t}\right)^2-k |\nabla u|^2 \right)dx dt$.
We shall illustrate what are the Euler--Lagrange equations when the Lagrangian
density depends on generalized fractional derivatives.
When we have the Lagrangian with the kinetic term depending
on the operator $B_{P_{t}}^{\alpha}$,
then the fractional action functional has the form
\begin{equation}
\label{ex:3}
\mathcal{J}(u)=\frac{1}{2} \int_{\Delta_3}\left[\rho \left(
B_{P_{t}}^{\alpha}u\right)^2-k |\nabla u|^2 \right]dx dt.
\end{equation}
The fractional Euler--Lagrange equation satisfied by an extremizer function of \eqref{ex:3} is
\begin{equation*}
\rho A_{P_{t}^{*}}^{\alpha} B_{P_{t}}^{\alpha}u-\nabla \cdot(k \nabla u)=0.
\end{equation*}
If $\rho$ and $k$ are constants, then the equation $\rho A_{P_{t}^{*}}^{\alpha}
B_{P_{t}}^{\alpha}u-c^2\Delta u=0$, $c^2=k/\rho$, can be called the
\emph{generalized time-fractional wave equation}. Now assume that the kinetic
and the potential energy depend on operators $B_{P_{t}}^{\alpha}$ and $B_{P}^{\beta}$,
$P=(P_{x_1},P_{x_2})$, $\beta=(\beta_1,\beta_2)$, respectively.
Then the action functional for the system has the form
\begin{equation}
\label{ex:2}
\mathcal{J}(u)=\frac{1}{2} \int_{\Delta_3}\left[\rho
\left(B_{P_{t}}^{\alpha}u\right)^2-k|\nabla_{B_{P}}^{\bm{\beta}} u|^2 \right]dx dt.
\end{equation}
The fractional Euler--Lagrange equation satisfied by an extremizer of \eqref{ex:2} is
\begin{equation*}
\rho A_{P_{t}^{*}}^{\alpha} B_{P_{t}}^{\alpha}u
-\sum\limits_{i=1}^{2} A_{P_{x_i}^{*}}^{\beta_i}(k B_{P_{x_i}}^{\beta_i}u)=0.
\end{equation*}
If $\rho$ and $k$ are constants, then $A_{P_{t}^{*}}^{\alpha} B_{P_{t}}^{\alpha}u
-c^2\left( \sum\limits_{i=1}^{2} A_{P_{x_i}^{*}}^{\beta_i} B_{P_{x_i}}^{\beta_i}u\right)=0$
can be called the \emph{generalized space- and time-fractional wave equation}.
\end{example}


\subsection{Dirichlet's Principle}
\label{sec:DP}

One of the most important variational principles for a PDE is Dirichlet's principle
for the Laplace equation. We shall present its generalized fractional counterpart.
In this section we assume that $N=1$. We show that the solution
of the generalized fractional boundary value problem
\begin{numcases}{ }
\sum_{i=1}^n A_{P_{t_i}^{*}}^{\alpha_i}\left(B_{P_{t_i}}^{\alpha_i}u\right)=0
& \text{ in } $\Delta_n$, \label{eq:BVP}\\
u=\psi & \text{ on } $\partial{\Delta_n}$, \label{b:BVP}
\end{numcases}
can be characterized as a minimizer of the energy functional
\begin{equation}
\label{eq:DirFunct}
\mathcal{J}[u]=\int\limits_{\Delta_n}
\left|\nabla_{B_{P}}^{\bm{\alpha}}u\right|^2 dt
\end{equation}
on the set
\begin{equation*}
\mathcal{A}=\left\{u\in C^1(\bar{\Delta}_n;\R):B_{P_{t_i}}^{\alpha_i}
u\in C^1(\bar{\Delta}_n;\R), \left.u\right|_{\partial \Delta_n}=\psi\right\},
\end{equation*}
where $\bm{\alpha}\in (0,1)^n$, $P=\left(P_{t_1},\dots, P_{t_n}\right)$,
$P^*=\left(P_{t_1}^*,\dots, P_{t_n}^*\right)$, and  $k_{1-\alpha_i}$
is a difference kernel such that $k_{1-\alpha_i}\in L_1\left((0,b_i-a_i);\R\right)$,
$i=1,\dots,n$.

\begin{remark}
In the following we assume that both problems, \eqref{eq:BVP}--\eqref{b:BVP}
and minimization of \eqref{eq:DirFunct} on the set $\mathcal{A}$, have solutions.
\end{remark}

\begin{theorem}[Generalized fractional Dirichlet's principle]
\label{thm:GFDP}
Let $\bm{\alpha}\in (0,1)^n$ and $u\in\mathcal{A}$.
Then $u$ solves the boundary value problem
\eqref{eq:BVP}--\eqref{b:BVP} if and only if $u$ satisfies
\begin{equation}
\label{eq:3}
\mathcal{J}[u]=\min\limits_{w\in \mathcal{A}}\mathcal{J}[w].
\end{equation}
\end{theorem}

\begin{proof}
Multiply the equation \eqref{eq:BVP} by any $v\in C^1(\bar{\Delta}_n;\R)$
such that $\left.v\right|_{\partial_{\Delta_n}}= 0$ and $B_{P_{t_i}}^{\alpha_i}v$
is continuously differentiable on the rectangle $\bar{\Delta}_n$. Then,
after integration,
\begin{equation*}
\int\limits_{\Delta_n}v\sum_{i=1}^n
A_{P_{t_i}^{*}}^{\alpha_i}\left(B_{P_{t_i}}^{\alpha_i}u\right)dt = 0.
\end{equation*}
The generalized integration by parts formula in Theorem~\ref{theorem:IBPD} yields
\begin{equation}\label{eq:1}
\int\limits_{\Delta_n}\nabla_{B_{P}}^{\bm{\alpha}}u
\cdot\nabla_{B_{P}}^{\bm{\alpha}}v dt = 0,
\end{equation}
as there is no boundary term since $\left.v\right|_{\partial_{\Delta_n}}= 0$.
By \eqref{eq:1} and properties of the scalar product, one has
\begin{equation*}
\begin{split}
\int\limits_{\Delta_n} \left|\nabla_{B_{P}}^{\bm{\alpha}}(u+v)\right|^2 dt
&=\int\limits_{\Delta_n} \left|\nabla_{B_{P}}^{\bm{\alpha}}u\right|^2 dt
+2\int\limits_{\Delta_n}\nabla_{B_{P}}^{\bm{\alpha}}u
\cdot \nabla_{B_{P}}^{\bm{\alpha}}v~dt
+\int\limits_{\Delta_n} \left|\nabla_{B_{P}}^{\bm{\alpha}}v\right|^2 dt\\
&\geq \int\limits_{\Delta_n} \left|\nabla_{B_{P}}^{\bm{\alpha}}u\right|^2 dt.
\end{split}
\end{equation*}
Conversely, if $u$ satisfies \eqref{eq:3}, then, by Theorem~\ref{thm:ELCaputo},
$u$ is a solution to \eqref{eq:BVP}--\eqref{b:BVP}.
\end{proof}

\begin{theorem}
\label{thm:unc:sol}
There exists at most one solution $u\in\mathcal{A}$
to problem \eqref{eq:BVP}--\eqref{b:BVP}.
\end{theorem}

\begin{proof}
Let $u\in \mathcal{A}$ be a solution to problem \eqref{eq:BVP}--\eqref{b:BVP}.
Assume that $\hat{u}$ is another solution to problem
\eqref{eq:BVP}--\eqref{b:BVP}. Then $w=u-\hat{u}\neq 0$ and
\begin{equation*}
\int_{\Delta_n}w\sum_{i=1}^n
A_{P_{t_i}^{*}}^{\alpha_i}\left(B_{P_{t_i}}^{\alpha_i}w\right)dt = 0.
\end{equation*}
By the generalized integration by parts formula (Theorem~\ref{theorem:IBPD}),
and since $\left.w\right|_{\partial\Delta_n}=0$, one has
\begin{equation*}
\int_{\Delta_n}\sum_{i=1}^n \left(B_{P_{t_i}}^{\alpha_i}w\right)^2 dt
=\int_{\Delta_n}\left|\nabla_{B_{P}}^{\bm{\alpha}}w\right|^2 dt = 0.
\end{equation*}
Note that $\left|\nabla_{B_{P}}^{\bm{\alpha}}w\right|^2$
is a nonnegative definite quantity. The volume integral of
a nonnegative definite quantity is equal to zero only in the case
when this quantity is zero itself throughout the volume. Thus,
$\nabla_{B_{P}}^{\bm{\alpha}}w=0$. Since $w$ is continuously
differentiable and $k_{1-\alpha_i}\in L_1\left((0,b_i-a_i);\R\right)$, we have
\begin{equation*}
\frac{\partial}{\partial t_i}w (t)\equiv 0, \quad i=1,\dots,n,
\end{equation*}
that is, $\nabla w=0$. Because $w=0$ on $\partial\Delta_n$,
we deduce that $w=0$. In other words, $u=\hat{u}$.
\end{proof}


\subsection{The Multidimensional Generalized Fractional Noether's Theorem}
\label{sec:NT}

Emmy Noether's theorem \cite{Noether} states that conservation laws in classical
mechanics follow whenever the Lagrangian function is invariant under
a one-parameter continuous group that transforms dependent and/or independent
variables \cite{torres:2002,torres:2004}.
In this section we prove a Noether-type theorem
for variational problems that depend on generalized partial fractional integrals
and derivatives. We start by introducing the notion of variational invariance.

\begin{definition}
\label{def:trans}
Functional \eqref{eq:FundFunct} is said to be invariant under an
$\varepsilon$-parameter family of infinitesimal transformations
\begin{equation}
\label{eq:trans}
\bar{u}(t)=u(t)+\varepsilon\xi(t,u(t))+o(\varepsilon)
\end{equation}
with $\xi\in C^1\left(\bar{\Delta}_n;\mathbb{R}^N\right)$ 
such that $B_{P_{t_i}^{1}}^{\alpha_i}\xi$
and $K_{P_{t_i}^{2}}^{\beta_i}\xi$ 
exist and are continuous on $\bar{\Delta}_n$,
$i\in\{1,\dots,n\}$, if
\begin{equation*}
\int\limits_{\Delta_n^*} F\left(t,u(t),\nabla_{B_{P^1}}^{\bm{\alpha}}u(t),
\nabla_{K_{P^2}}^{\bm{\beta}}u(t)\right)dt
=\int\limits_{\Delta_n^*} F\left(t,\bar{u}(t),
\nabla_{B_{P^1}}^{\bm{\alpha}}\bar{u}(t),
\nabla_{K_{P^2}}^{\bm{\beta}}\bar{u}(t)\right)dt
\end{equation*}
for any $\Delta_n^*\subseteq\Delta_n$.
\end{definition}

The following result provides a necessary condition of invariance.

\begin{lemma}
If functional \eqref{eq:FundFunct} is invariant under an $\varepsilon$-parameter
family of infinitesimal transformations \eqref{eq:trans}, then
\begin{multline}
\label{eq:NCI}
\sum\limits_{k=1}^{N}\Biggl(\partial_{n+k}
F\left\{u\right\}_{P^1, P^2}^{\bm{\alpha},\bm{\beta}}(t)\cdot \xi_k(t,u)
+\sum\limits_{i=1}^{n} \left[\partial_{N+kn+i}
F\left\{u\right\}_{P^1, P^2}^{\bm{\alpha},\bm{\beta}}(t)
B_{P_{t_i}^{1}}^{\alpha_i}\xi_k(t,u)\right.\\
\left.+\partial_{n+N(k+n)+i} F\left\{u\right\}_{P^1, P^2}^{\bm{\alpha},
\bm{\beta}}(t)K_{P_{t_i}^{2}}^{\beta_i}\xi_k(t,u)\right]\Biggr)=0.
\end{multline}
\end{lemma}

\begin{proof}
By Definition~\ref{def:trans}, invariance of functional \eqref{eq:FundFunct}
under transformations \eqref{eq:trans} is equivalent to
\begin{equation}
\label{inv_eq}
F\left(t,u,\nabla_{B_{P^1}}^{\bm{\alpha}}u,
\nabla_{K_{P^2}}^{\bm{\beta}}u\right)
= F\left(t,\bar{u},\nabla_{B_{P^1}}^{\bm{\alpha}}\bar{u},
\nabla_{K_{P^2}}^{\bm{\beta}}\bar{u}\right).
\end{equation}
Let us differentiate \eqref{inv_eq} with respect to $\varepsilon$:
\begin{multline*}
\frac{d}{d\varepsilon}F\Biggl(t,u(t)+\varepsilon\xi(t,u(t))+o(\varepsilon),
\nabla_{B_{P^1}}^{\bm{\alpha}}\left(u(t)+\varepsilon\xi(t,u(t))
+o(\varepsilon)\right),\\
\nabla_{K_{P^2}}^{\bm{\beta}}\left(u(t)
+\varepsilon\xi(t,u(t))+o(\varepsilon)\right)\Biggr) = 0.
\end{multline*}
Putting $\varepsilon=0$ and applying definitions and properties
of partial generalized fractional operators, we obtain \eqref{eq:NCI}.
\end{proof}

In order to state Noether's theorem in a compact form, we follow \cite{gastao1}.
More precisely, we introduce two bilinear operators.

\begin{definition}
\label{def:1}
Let $f,g\in C^1(\bar{\Delta}_n;\R)$ for $K_{P_{t_i}^{1}}^{1-\alpha_i}
g\in C^1(\bar{\Delta}_n;\R)$. We define the following bilinear operators:
\begin{equation*}
\begin{split}
\textsl{D}^{\alpha_i}_{P_{t_i}^1}[f,g]
&:=f A_{P_{t_i}^{1*}}^{\alpha_i}g+gB_{P_{t_i}^{1}}^{\alpha_i}f,\\
\textsl{I}^{\beta_i}_{P_{t_i}^2}[f,g]
&:=-f K_{P_{t_i}^{2*}}^{\beta_i}g+gK_{P_{t_i}^{2}}^{\beta_i}f,
\end{split}
\end{equation*}
$i=1\dots,n$.
\end{definition}

Now we are ready to state our generalized fractional Noether's theorem.

\begin{theorem}[Multidimensional generalized fractional Noether's theorem]
\label{noether}
If functional \eqref{eq:FundFunct} is invariant,
in the sense of Definition~\ref{def:trans}, then
\begin{multline}
\label{eq:Noether}
\sum\limits_{k=1}^N\sum\limits_{i=1}^n\Biggl[\textsl{D}^{\alpha_i}_{P_{t_i}^1}
[\xi_k(t,u(t)),\partial_{N+kn+i} F\left\{u\right\}_{P^1, P^2}^{\bm{\alpha},\bm{\beta}}(t)]\\
+\textsl{I}^{\beta_i}_{P_{t_i}^2}[\xi_k(t,u(t)),\partial_{n+N(k+n)+i}
F\left\{u\right\}_{P^1, P^2}^{\bm{\alpha},\bm{\beta}}(t)]\Biggr]=0
\end{multline}
along any extremal of \eqref{eq:FundFunct}.
\end{theorem}

\begin{proof}
By equations \eqref{eq:eqELCaputo} we have
\begin{multline}
\label{eq:5}
\partial_{n+k} F\left\{u\right\}_{P^1, P^2}^{\bm{\alpha},\bm{\beta}}(t)
= \sum_{i=1}^n \left[A_{P_{t_i}^{1*}}^{\alpha_i}\partial_{N+kn+i}
F\left\{u\right\}_{P^1, P^2}^{\bm{\alpha},\bm{\beta}}(t)\right.\\
\left.-K_{P_{t_i}^{2*}}^{\beta_i}\partial_{n+N(k+n)+i}
F\left\{u\right\}_{P^1, P^2}^{\bm{\alpha},\bm{\beta}}(t)\right],
\quad k=1,\dots,N.
\end{multline}
Putting \eqref{eq:5} into \eqref{eq:NCI}, we obtain that
\begin{multline*}
\sum\limits_{k=1}^N\sum\limits_{i=1}^n \Biggl[
\xi_k(t,u(t))A_{P_{t_i}^{1*}}^{\alpha_i}\partial_{N+kn+i}
F\left\{u\right\}_{P^1, P^2}^{\bm{\alpha},\bm{\beta}}(t)\\
-\xi_k(t,u(t))K_{P_{t_i}^{2*}}^{\beta_i}\partial_{n+N(k+n)+i}
F\left\{u\right\}_{P^1, P^2}^{\bm{\alpha},\bm{\beta}}(t)
+\partial_{N+i+kn}F\left\{u\right\}_{P^1, P^2}^{\bm{\alpha},\bm{\beta}}(t)
B_{P_{t_i}^{1}}^{\alpha_i}\xi_k(t,u(t))\\
+\partial_{n+N(k+n)+i}F\left\{u\right\}_{P^1, P^2}^{\bm{\alpha},\bm{\beta}}(t)
K_{P_{t_i}^{2}}^{\beta_i}\xi_k(t,u(t))\Biggr]=0.
\end{multline*}
Finally, we arrive to \eqref{eq:Noether} by Definition~\ref{def:1}.
\end{proof}

\begin{example}
Let $N=1$, $\bm{\alpha},\bm{\beta}\in (0,1)^n$, $c\in\mathbb{R}$
and $P=(P_{t_1},\dots,P_{t_n})$ with $P_{t_i}=\langle a_i,t_i,b_i,p_i,q_i\rangle$,
$i=1,\dots, n$. Consider the $\varepsilon$-parameter family of infinitesimal transformations
\begin{equation}
\label{eq:tranex}
\bar{u}(t)=u(t)+\varepsilon c+o(\varepsilon)
\end{equation}
and the functional
\begin{equation*}
\mathcal{J}[u]=\int\limits_{\Delta_n}F\left(t,\nabla_{B_{P}}^{\bm{\alpha}}u(t)\right)dt.
\end{equation*}
Then, for any $\Delta_n^{*}\subseteq\Delta_n$, we have
\begin{equation*}
\int\limits_{\Delta_n^{*}}F\left(t,\nabla_{B_{P}}^{\bm{\alpha}}\bar{u}(t)\right)dt
=\int\limits_{\Delta_n^{*}}F\left(t,\nabla_{B_{P}}^{\bm{\alpha}}u(t)\right)dt.
\end{equation*}
Hence, $\mathcal{J}[u]$ is invariant under transformations \eqref{eq:tranex}
and Theorem~\ref{noether} asserts that
\begin{equation*}
\sum\limits_{i=1}^{n}\textsl{D}^{\alpha_i}_{P_{t_i}^1}\left[c,
\partial_{n+i} F \left(t,\nabla_{B_{P}}^{\bm{\alpha}}u(t)\right)\right]=0.
\end{equation*}
\end{example}


\section{Conclusion}
\label{sec:conc}

Partial fractional integrals and derivatives can be defined in different ways and, consequently,
in each case one must consider different variational problems. In this paper
we unify and extend previous results of the multidimensional calculus of variations
by considering more general operators that reduce to the standard
fractional integrals and derivatives by an appropriate choice of kernels
and $p$-sets. After proving generalized integration by parts formulas,
we obtained Euler--Lagrange equations, a generalized fractional Dirichlet's principle,
and a fractional Noether's theorem. As an example, we obtained
a generalized space- and time-fractional wave equation.

This paper marks the born of the generalized multidimensional fractional calculus of variations.
Much remains to be done. For example, if boundary conditions are not imposed at the initial problem,
then Theorem~\ref{thm:ELCaputo} needs to be complemented with transversality conditions. Problems
subject to constraints can also be considered.


\section*{Acknowledgements}

This work was supported by FEDER funds through
COMPETE --- Operational Programme Factors of Competitiveness
(``Programa Operacional Factores de Competitividade'')
and by Portuguese funds through the
Center for Research and Development
in Mathematics and Applications (University of Aveiro)
and the Portuguese Foundation for Science and Technology
(``FCT --- Funda\c{c}\~{a}o para a Ci\^{e}ncia e a Tecnologia''),
within project PEst-C/MAT/UI4106/2011
with COMPETE number FCOMP-01-0124-FEDER-022690.
Odzijewicz was also supported by FCT through the
PhD fellowship SFRH/BD/33865/2009;
Malinowska by Bialystok University of Technology
grant S/WI/02/2011; Odzijewicz and Torres by EU funding
under the 7th Framework Programme FP7-PEOPLE-2010-ITN,
grant agreement number 264735-SADCO.




\begin{thebibliography}{xx}

\bibitem{Abel}
N.H. Abel,
\textit{Euvres completes de Niels Henrik Abel}
(Christiana: Imprimerie de Grondahl and Son; New York and London:
Johnson Reprint Corporation. VIII, 621 pp., 1965)

\bibitem{OmPrakashAgrawal}
O.P. Agrawal,
Comput. Math. Appl. \textbf{59}, 1852 (2010)

\bibitem{MyID:182}
R. Almeida, A.B. Malinowska, D.F.M. Torres,
J. Math. Phys. \textbf{51}, 033503 (2010)
{\tt arXiv:1001.2722}

\bibitem{DerInt}
R. Almeida, D.F.M. Torres,
Appl. Math. Lett. \textbf{22}, 1816 (2009)
{\tt arXiv:0907.1024}

\bibitem{Baleanu}
D. Baleanu, S. Muslih,
Physica Scripta \textbf{72}, 119 (2005)

\bibitem{B:O:T}
L. Bourdin, T. Odzijewicz, D.F.M. Torres,
Adv. Dyn. Syst. Appl. \textbf{8}, 3 (2013)
{\tt arXiv:1208.2363}

\bibitem{Cresson}
J. Cresson,
J. Math. Phys. \textbf{48}, 033504 (2007)

\bibitem{book:Evans}
L.C. Evans,
\textit{Partial Differential Equations}
(Gruaduate Studies in Mathematics,
American Mathematical Society, United States of America, 1997)

\bibitem{book:Ewing}
G.M. Ewing,
\textit{Calculus of variations with applications}
(Courier Dover Publications, New York, 1985)

\bibitem{gastao1}
G.S.F. Frederico, D.F.M. Torres,
J. Math. Anal. Appl. \textbf{334}, 834 (2007)
{\tt arXiv:math/0701187}

\bibitem{gastao4}
G.S.F. Frederico, D.F.M. Torres,
Appl. Math. Comput. \textbf{217}, 1023 (2010)
{\tt arXiv:1001.4507}

\bibitem{book:Jurgen}
J. Jost, X. Li-Jost,
\textit{Calculus of variations}
(Cambridge Univ. Press, Cambridge, 1998)

\bibitem{book:Kilbas}
A.A. Kilbas, H.M. Srivastava, J.J. Trujillo,
\textit{Theory and applications of fractional differential equations}
(North-Holland Mathematics Studies, 204, Elsevier, Amsterdam, 2006)

\bibitem{book:Kiryakova}
V. Kiryakova,
\textit{Generalized fractional calculus and applications}
(Longman Sci. Tech., Harlow, 1994)

\bibitem{book:Klimek}
M. Klimek,
\textit{On solutions of linear fractional differential equations of a variational type}
(The Publishing Office of Czenstochowa University of Technology, Czestochowa, 2009)

\bibitem{Lupa}
M. Klimek, M. Lupa,
Fract. Calc. Appl. Anal. \textbf{16}, 243 (2013)

\bibitem{MyID:249}
M.J. Lazo, D.F.M. Torres,
J. Optim. Theory Appl. \textbf{156}, 56 (2013)
{\tt arXiv:1210.0705}

\bibitem{malinowska:1}
A.B. Malinowska,
Appl. Math. Lett. \textbf{25}, 1941 (2012)
{\tt arXiv:1203.2107}

\bibitem{malinowska:2}
A.B. Malinowska, 
J. Vib. Control \textbf{19}, 1161 (2013)
{\tt arXiv:1203.2102}

\bibitem{book:fcv}
A.B. Malinowska, D.F.M. Torres,
\textit{Introduction to the fractional calculus of variations}
(Imperial College Press, London \&
World Scientific Publishing, Singapore, 2012).

\bibitem{Noether}
E. Noether,
G\"{o}tt. Nachr., 235 (1918)

\bibitem{MyID:207}
T. Odzijewicz, A.B. Malinowska, D.F.M. Torres,
Nonlinear Anal. \textbf{75}, 1507 (2012)
{\tt arXiv:1101.2932}

\bibitem{FVC_Gen}
T. Odzijewicz, A.B. Malinowska, D.F.M. Torres,
Comput. Math. Appl. \textbf{64}, 3351 (2012)
{\tt arXiv:1201.5747}

\bibitem{FVC_Gen_Int}
T. Odzijewicz, A.B. Malinowska, D.F.M. Torres,
Abstr. Appl. Anal., ID 871912 (2012)
{\tt arXiv:1203.1961}

\bibitem{OMT:2013}
T. Odzijewicz, A.B. Malinowska, D.F.M. Torres,
Fract. Calc. Appl. Anal. \textbf{16}, 64 (2013)
{\tt arXiv:1205.4851}

\bibitem{tatiana}
T. Odzijewicz, D.F.M. Torres,
Balkan J. Geom. Appl. \textbf{16}, 102 (2011)
{\tt arXiv:1102.1337}

\bibitem{book:Podlubny}
I. Podlubny,
\textit{Fractional differential equations}
(Academic Press, San Diego, CA, 1999)

\bibitem{CD:Riewe:1996}
F. Riewe,
Phys. Rev. E (3) \textbf{53}, 1890 (1996)

\bibitem{CD:Riewe:1997}
F. Riewe,
Phys. Rev. E (3) \textbf{55}, 3581 (1997)

\bibitem{book:Samko}
S.G. Samko, A.A. Kilbas, O.I. Marichev,
\textit{Fractional integrals and derivatives}
(Gordon and Breach, Yverdon, 1993)

\bibitem{tarasov}
V.E. Tarasov,
Ann. Phys. \textbf{323}, 2756 (2008)

\bibitem{torres:2002}
D.F.M. Torres,
Eur. J. Control \textbf{8}, 56 (2002)

\bibitem{torres:2004}
D.F.M. Torres,
Commun. Pure Appl. Anal. \textbf{3}, 491 (2004)

\bibitem{book:vanBrunt}
B. van Brunt,
\textit{The calculus of variations}
(Universitext, Springer, New York, 2004)

\bibitem{book:Weinstock}
R. Weinstock,
\textit{Calculus of variations with applications to physics and engineering}
(McGraw Hill Book Company Inc., 1952)

\end{thebibliography}
\end{document}